\documentclass[aps,pra,groupedaddress,showpacs,twocolumn,superscriptaddress]{revtex4-1}
 
\usepackage{amsmath}
\usepackage{amssymb}
\usepackage{amsthm}
\usepackage[hyperindex,breaklinks]{hyperref}
\usepackage{bm}
\usepackage{graphicx}
\usepackage{graphics}
\usepackage{mathrsfs}
\usepackage{tikz}

\usepackage{xcolor}

\newcommand{\FF}{\mathbb{F}}

\newcommand{\HC}{\mathscr{H}}

\newcommand{\ket}[1]{|#1\rangle}                  
\newcommand{\bra}[1]{\left\langle #1 \right|}     
\newcommand{\dyad}[2]{\ket{#1}\bra{#2}}           
\newcommand{\Tr}{{\rm Tr}}                        
\newcommand{\vect}[1]{\bm{#1}}         

\newcommand{\ii}{\mathrm{i}}					  

\def\dya#1{|#1\rangle \langle#1|}

\long\def\ca#1\cb{} 

\newtheorem{theorem}{Theorem}
\newtheorem{lemma}[theorem]{Lemma}

\newtheorem*{definition*}{Definition}

\newcommand{\pt}{A}
\newcommand{\ptc}{B}

\begin{document}

\title{Accessing quantum secrets via local operations and classical communication}

\author{Vlad Gheorghiu}
\email{vgheorgh@gmail.com}
\affiliation{Institute for Quantum Science and Technology,
University of Calgary,
Calgary, AB, T2N 1N4, Canada}
\affiliation{Department of Mathematics and Statistics,
University of Calgary,
Calgary, AB, T2N 1N4, Canada}
\affiliation{Institute for Quantum Computing, University of Waterloo, Waterloo, ON, N2L 3G1}

\author{Barry C. Sanders}
\affiliation{Institute for Quantum Science and Technology,
University of Calgary,
Calgary, AB, T2N 1N4, Canada}
\affiliation{Hefei National Laboratory for Physical Sciences at Microscale,\textcolor{black}
University of Science and Technology of China, P.\ R.\ China}

\date{Version of August 30, 2013}

\begin{abstract}
Quantum secret-sharing  and quantum error-correction schemes rely on multipartite decoding protocols, yet the non-local operations involved are challenging and sometimes infeasible. Here we construct a quantum secret-sharing protocol with a reduced number of quantum communication channels between the players. We introduce a scheme based on embedding a classical linear code into a quantum error-correcting code, then mapping the latter to a quantum secret-sharing protocol. \textcolor{black}{In contrast to the Calderbank-Shor-Steane construction,
 we do not impose any restriction on the classical code; our protocol works with any arbitrary linear code.} Our work paves the way towards the more general problem of simplifying the decoding of quantum error-correcting codes.  
\end{abstract}

\maketitle

\section{Introduction}\label{sct1}
Secret sharing is a cryptographic protocol in which a dealer distributes a shared secret among a set of players, so that only certain authorized subsets can collaboratively recover the secret.  The protocol, first introduced by Shamir \cite{Shamir:1979} and Blakley \cite{Blakley}, is important in any area that requires sharing of highly sensitive information, such as bank accounts, missile launch sequences etc. 

The quantum counterpart is a scheme in which the dealer distributes either a classical secret \cite{PhysRevA.59.1829} (string of bits) or a quantum secret (quantum state) \cite{PhysRevLett.83.648} to the set of players via quantum channels \cite{PhysRevA.78.042309,PhysRevA.82.062315}. Quantum secret-sharing is useful for distributing shared quantum keys, non-counterfeitable  ``quantum money" \cite{WiesnerQMoney}, distributed quantum computing \cite{4031361}, secure quantum memory  and multipartite quantum communication \cite{PhysRevLett.89.097905}. For the quantum secret-sharing protocol to be feasible, the dealer is assumed to be ``powerful" -- she can prepare arbitrary quantum states and reliably distribute them to the players. The players have full access to universal quantum computers and can communicate among themselves via quantum channels so that only certain authorized subsets can recover (decode) the secret. 
  The decoding operation is harder to implement than in the classical case, as it requires quantum communication which is expensive.

Reducing the amount of quantum communication required for the decoding can improve the efficiency of distributed cryptographic protocols in which a subset of the players have restricted communication capabilities. Consider for example a quantum secret-sharing scheme with players divided into two subsets, one of which is computationally powerful (each player has access to universal quantum computation and all players can use quantum communication), whereas the other one is computationally weak (each player has access to local universal quantum computers but the players can use only classical communication between them). One such instance is a secret-sharing scheme between Earth (the computationally powerful subset) and, for example, the International Space Station (the computationally weak subset). 

Reducing the amount of quantum communication (i.e. reducing the number of non-local operations involved) also helps simplify the decoding of quantum error-correcting codes \cite{PhysRevA.52.R2493,PhysRevA.54.1098}, which are of crucial importance for the construction of a real-world fault-tolerant quantum computer. 

In this article we solve the following problem. For a large class of quantum secret-sharing schemes constructed from classical linear error-correcting codes \cite{MacWilliamsSloane:TheoryECC}, we show that their decoding can be simplified by replacing some of the quantum channels among the players by classical ones. Inspired by the Calderbank-Shor-Steane (CSS) construction \cite{NielsenChuang:QuantumComputation} we embed a classical linear error-correcting code into a quantum code and then show that this embedding induces a quantum secret-sharing scheme in which all players have to collaborate to recover the secret. In this protocol some of the players are only required to perform local measurements and share their measurement results via classical channels.
\textcolor{black}{In contrast to the CSS construction, which uses two classical linear codes (with the restriction that one is weakly dual to the other~\cite{NielsenChuang:QuantumComputation}),
our protocol works with any \emph{arbitrary} classical linear error-correcting code, with no restriction whatsoever.}

\textcolor{black}{The remainder of this article is organized as follows. In Sec.~\ref{sct2} we introduce our secret-sharing protocol, followed by describing the decoding operation, and then introduce ``optimal schemes". 
Section~\ref{sct3} presents an illustrative example, and we conclude in Sec.~\ref{sct4}.}

\section{The secret sharing protocol}\label{sct2}
\subsection{Embedding a classical code into a quantum subspace}
We begin by considering an $[n,k,d]_q$ classical error-correcting code over $\FF_q$, the Galois field with 
\mbox{$q=p^m$}
elements, where $p$ is prime and $m$ is a positive integer. The parameter $k$ denotes the number of encoded \emph{dits} (generalization of a bit that allows holding more than 2 states), $n$ is the number of carriers and $d$ is the distance of the code. We can represent such a code compactly using a $k\times n$ \emph{generator matrix} $G$ with elements in $\FF_q$. Each codeword ($n$-tuple in $\FF_q^n$ \footnote{Throughout the paper we represent tuples as row vectors, and use the transpose symbol $T$ to denote a column vector.}) can then be written as  
\begin{equation}\label{eqn1}
\vect{x}\cdot G\:=\sum_{ij}x_iG_{ij},
\end{equation}
where $\vect{x}$ is a $k$-tuple in $\FF_q^k$,
for a total number of codewords equal to $q^k$, where the addition and multiplication in \eqref{eqn1} are over the finite field $\FF_q$.  One can regard $G$ as a linear mapping from the ``input" space $\FF_q^k$ to the ``output" (or encoded) subspace of $\FF_q^n$, see the top of our diagram (mapping 1) in Fig.~\ref{fgr1}. 

We use the elements $\vect{x}\in \FF_q^k$ to label the basis vectors of $\HC^{\otimes k}$, the Hilbert space of $k$ qudits, and denote the collection of the orthonormal basis vectors by $\{\ket{\vect{x}}\}_{\vect{x}\in\FF_q^k}$  (see the mapping 4 in Fig.~\ref{fgr1}). Similarly we embed the elements $\vect{x}\cdot G\in\FF_q^n$ into a subspace of the $\HC^{\otimes n}$ spanned by the collection of orthonormal vectors $\{\ket{\vect{x}\cdot G}\}_{\vect{x}\in\FF_q^k}$, as depicted by mapping 2 in Fig.~\ref{fgr1}. Note that $\HC^{\otimes k}$ is isomorphic with Span$\{\ket{\vect{x}\cdot G}\}_{\vect{x}\in\FF_q^k}$ through an encoding isometry $V$, see the bottom of our diagram (mapping 3) in Fig.~\ref{fgr1}. In particular, the isometry $V$ can be explicitly constructed from the generating matrix $G$ using a simple quantum circuit that consists of controlled-NOT gates; see Sec.~10.5.8 of Ref.~\cite{NielsenChuang:QuantumComputation}. 

\begin{figure}
\centering
\includegraphics{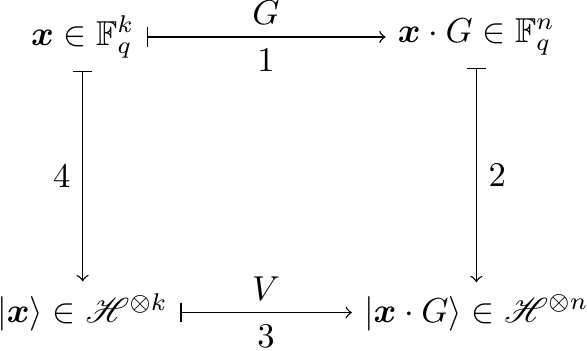}
\caption{A schematic for the various embeddings used throughout the article.}
\label{fgr1}
\end{figure}

We have all the ingredients to construct a quantum secret-sharing scheme as follows.
A dealer holds a $k$-qudit quantum secret 
\begin{equation}\label{eqn2}
\ket{\psi}=\sum_{\vect{x}\in \FF_q^k}c(\vect{x})\ket{\vect{x}},
\end{equation}
with $c(\vect{x})$ normalized complex coefficients. The secret is then distributed to a set of $n$ players using the isometric encoding $V$, so the state shared by the players is 
\begin{equation}\label{eqn3}
\ket{\Psi}=\sum_{\vect{x}\in \FF_q^k}c(\vect{x}) \ket{\vect{x}\cdot G}.
\end{equation}

\subsection{Decoding via local operations and classical communication}
We next design a decoding protocol in which all $n$ players have to collaborate; however, just a proper subset $\pt$ of the entire set of players $P$ is required to use local operations and classical communication (LOCC) with the complementary subset $\ptc$. The latter subset can then fully recover the quantum secret.
Our scheme is depicted in Fig.~\ref{fgr2}. 
Our secret-sharing scheme is imperfect (or ``ramp"); i.e., there exist subsets of players that may extract partial information about the secret. However, we can transform it to a perfect (or ``threshold") quantum secret-sharing scheme via ``twirling" and allowing the dealer to share extra classical communication channels with the players \cite{6225432,PhysRevA.85.052309}.
\begin{figure}
\includegraphics[scale=0.35]{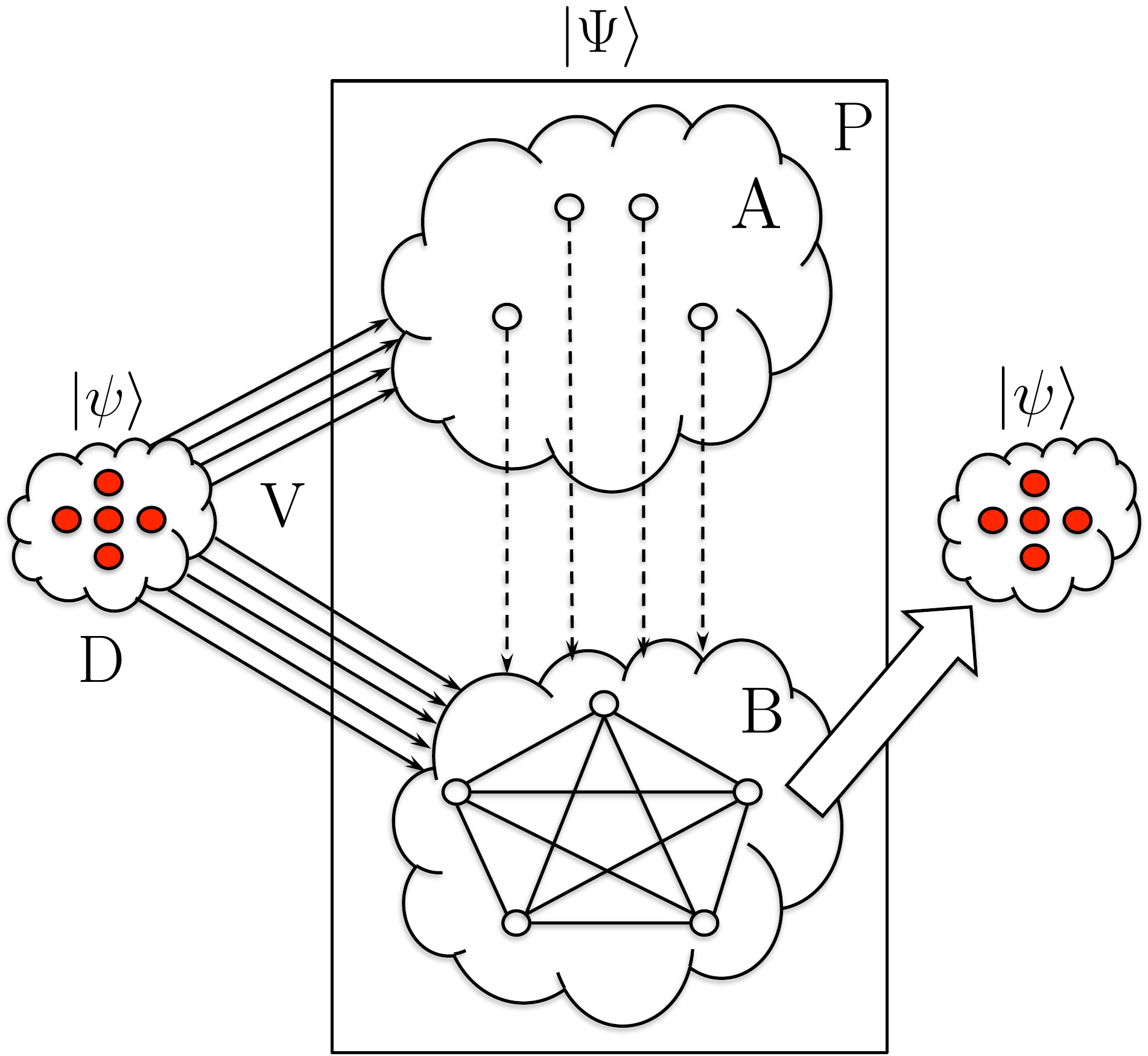}
\caption{(Color online) LOCC-recoverable quantum secret-sharing. The dealer $D$ encodes $k$ qudits (red-filled circles) in state $\ket{\psi}$ via the isometry $V$  to realize the state $\ket{\Psi}$ and then transmits it to the players $P$ (with each player denoted by an empty circle) via quantum channels (solid lines indicate single-qudit channels). The players in $A$ perform local measurements and each one communicates via classical channels (dashed lines) with all players in $B$, who are all connected via quantum channels. Finally the players in $B$ perform a global quantum operation to recover $\ket{\psi}$.}
\label{fgr2}
\end{figure}

\begin{definition*}\label{dfn1}
A subset $\pt$ of the entire set of players $P$ is \emph{LOCC-assisting} for its complement $\ptc$ whenever there exists an LOCC scheme that $\pt$ can perform, followed by sending the measurement results to $\ptc$, so that $\ptc$ can fully recover the quantum secret.
\end{definition*}

Inspired by the CSS construction, we employ the concept of embedding a classical code into a quantum code, but use the technique to construct a differrent way of decoding quantum secret-sharing schemes. The most important outcome of our scheme is a drastic reduction of the number of interplayer quantum communication channels required for the decoding. Our main result is summarized below.

\begin{theorem}\label{thm1}
Let $[n,k,d]_q$ be a classical error-correcting code with generator matrix $G$, and let $\ket{\Psi}$ be a $k$-qudit quantum secret distributed to a set of $n$ players using the isometric encoding
\begin{equation}\label{eqn4}
\ket{\vect{x}}\stackrel{V}{\longmapsto}\ket{\vect{x}\cdot G}.
\end{equation}
Let $\pt$ be a subset of the carrier qudits, and let $\ptc$ denote its complement.  Let $G_\ptc$ be the matrix obtained by removing the columns that correspond to the players in $\ptc$ from $G$. Then the subset $\pt$ is LOCC-assisting for its complement $\ptc$
\emph{if and only if}
\begin{equation}\label{eqn5}
\mathrm{rank}(G_{\ptc}) = k.
\end{equation}
\end{theorem}

\begin{proof}
Consider that each player in $\pt$ performs a local measurement in the Fourier basis $\{\ket{\bar x}=F\ket{x}\}_{x\in \FF_q}$, where $F$ is the generalized Fourier matrix defined as
\begin{equation}\label{eqn6}
F:=\frac{1}{\sqrt{q}}\sum_{x,z\in \FF_q}\omega^{\mathrm{tr}(xz)}\dyad{z}{x}, \quad\omega:=\exp(2\pi\ii/p), 
\end{equation}
where $\mathrm{tr}(x)$ denotes the ``trace" \cite{quantph.9705052,doi:10.1142/S0129054103002011} of an element $x\in \FF_{q=p^m}$,
\begin{equation}\label{eqn7}
\mathrm{tr}:\FF_q\longrightarrow\FF_p,\quad x\stackrel{\mathrm{tr}}{\longmapsto} \sum_{i=0}^{m-1} x^{p^i}\in\FF_p.
\end{equation}
We denote by $a_k\in \FF_q$ the label of the measurement result of the $k$th player. For compactness we collect all measurement results that the players in $\pt$ perform into a vector $\vect{a}\in \FF_q^{|\pt|}$, where $|\pt|$ represents the number of players in $\pt$. Let 
\begin{equation}\label{eqn8}
\ket{\Psi}_{\ptc,\vect{a}}:=\frac{\Tr_{\pt}\left[(\dya{\vect{a}}\otimes I_{\ptc})\ket{\Psi}\right])}
{\lVert\Tr_{\pt}\left[(\dya{\vect{a}}\otimes I_{\ptc})\ket{\Psi}\right])\rVert}
\end{equation}
be the normalized resultant state of the remaining $\ptc$ players, given that the results of the measurements by $\pt$ are $\vect{a}$. 
From \eqref{eqn6} all measurement results have the same probability, independent of the secret $\ket{\psi}$ in \eqref{eqn2}, 
\begin{equation}\label{eqn9}
p(\vect{a})=\lVert\Tr_{\pt}\left[(\dya{\vect{a}}\otimes I_{\ptc})\ket{\Psi}\right])\rVert^2=\frac{1}{q^{|\pt|}}.
\end{equation}
The resultant state on $\ptc$, given the measurement result $\vect{a}$, is
\begin{equation}\label{eqn10}
\ket{\Psi}_{\ptc,\vect{a}}=\sum_{\vect{x}\in \FF_q^k}c(\vect{x})\omega^{-\mathrm{tr}(\vect{x}\cdot G_{\pt}\cdot \vect{a}^T)}\ket{\vect{x}\cdot G_\ptc}.
\end{equation}
Note that rank$(G_{\ptc})\leqslant k$, as $G_{\ptc}$ is obtained from the rank-$k$ generator matrix $G$ by removing columns from the latter.

If rank$(G_{\ptc})<k$, the number of mutually orthogonal states in \eqref{eqn10} is less than the dimension $q^k$ of the quantum secret $\ket{\psi}$ in \eqref{eqn2}, or, equivalently, dim(span($\{\ket{\vect{x}\cdot G_{\ptc}}\}_{\vect{x}\in \FF_q^k}))<$dim(span($\{\ket{\vect{x}}\}_{\vect{x}\in\FF_q^k}))$. In this case it is impossible to map the state $\ket{\Psi}_{\ptc,\vect{a}}$ in $\eqref{eqn10}$ back to $\ket{\psi}$ by an isometry that does not depend on the coefficients $c(\vect{x})$; there is not enough ``space" to ``fit" the secret $\ket{\psi}$ in the state $\ket{\Psi}_{\ptc,\vect{a}}$ and information is irreversibly lost \cite{PhysRevA.61.042311, PhysRevA.76.062320}. 

On the other hand the vectors $\{\ket{\vect{x}\cdot G_{\ptc}}\}_{\vect{x}\in \FF_q^k}$ are mutually orthogonal if and only if rank$(G_{\ptc}) = k$. In this latter case, the state $\ket{\Psi}_{\ptc,\vect{a}}$ \emph{can} be mapped back to the original secret $\ket{\psi}$ via a \emph{decoding isometry} (that depends on the measurement results $\vect{a}$) defined via
\begin{equation}\label{eqn11}
\ket{\vect{x}\cdot G_{\ptc}} \mapsto \omega^{\mathrm{tr}(\vect{x}\cdot G_{\pt}\cdot \vect{a}^T)}\ket{\vect{x}}.
\end{equation}
\end{proof}
Our next result shows how to construct the above decoding isometry explicitly.
\begin{theorem}\label{thm2}
Let $\pt$ be an  LOCC-assisting subset for its complement $\ptc$. Then the decoding isometry for $\ptc$ is a product of a local unitary operation, which depends only on the measurement results $\vect{a}$, and an isometry that depends only on the subset $\ptc$. A corresponding decoding quantum circuit can be constructed explicitly.
\end{theorem}
\begin{proof}
For some $\vect{z}\in \FF_q^{n-|A|}$ consider the action of $Z^{\vect{z}}:=Z^{Z^1}\otimes\cdots\otimes Z^{Z^{n-|A|}}$ on $\ket{\Psi}_{\ptc, \vect{a}}$ in \eqref{eqn10}, where $Z^z$ is the generalized single-qudit Weyl-Heisenberg operator \cite{quantph.9705052,doi:10.1142/S0129054103002011} defined as
\begin{equation}\label{eqn12}
Z^z:=\sum_{x\in \FF_q}\omega^{\mathrm{tr}(zx)}\dya{x},\text{ for }z\in \FF_q.
\end{equation}
The resultant state is
\begin{align}\label{eqn13}
Z^{\vect{z}}& \ket{\Psi}_{\ptc, \vect{a}}= \sum_{\vect{x}\in \FF_q^k}c(\vect{x})\omega^{-\mathrm{tr}(\vect{x}\cdot G_{\pt}\cdot \vect{a}^T)}Z^{\vect{z}}\ket{\vect{x}\cdot G_\ptc}\notag\\
&=\sum_{\vect{x}\in \FF_q^k}c(\vect{x})\omega^{-\mathrm{tr}(\vect{x}\cdot G_{\pt}\cdot \vect{a}^T)}\omega^{\mathrm{tr}(\vect{x}\cdot G_{\ptc}\cdot\vect{z}^T)}\ket{\vect{x}\cdot G_\ptc}\notag\\
&=\sum_{\vect{x}\in \FF_q^k}c(\vect{x})\omega^{\mathrm{tr}\left[\vect{x}\cdot (G_{\ptc}\cdot\vect{z}^T-G_{\pt}\cdot \vect{a}^T)\right]}\ket{\vect{x}\cdot G_\ptc}.
\end{align}
We now claim that there always exists a $\vect{z}\in \FF_q^{n-|\pt|}$ such that 
\begin{equation}\label{eqn14}
G_{\ptc}\cdot\vect{z}^T=G_{\pt}\cdot \vect{a}^T.
\end{equation}
As $\pt$ is LOCC-assisting for $\ptc$, Theorem~\ref{thm1} implies that rank$(G_{\ptc})=k$.  This fact implies that \eqref{eqn14} admits at least one solution $\vect{z}$ which can be found by elementary linear algebra methods over finite fields. Therefore, for such a solution $\vect{z}$, the operator $Z^{\vect{z}}$ eliminates all phases in \eqref{eqn10}
\begin{equation}\label{eqn15}
Z^{\vect{z}} \ket{\Psi}_{\ptc, \vect{a}}= \sum_{\vect{x}\in \FF_q^k}c(\vect{x})\ket{\vect{x}\cdot G_\ptc},
\end{equation}
which implies at once that the resultant state \eqref{eqn15} can be mapped back to the original secret \eqref{eqn2} by an isometry $V_{\ptc}$ defined by 
\begin{equation}\label{eqn16}
\ket{\vect{x}\cdot G_{\ptc}}\stackrel{V_{\ptc}}{\longmapsto} \ket{\vect{x}}.
\end{equation}
The overall recovery procedure can be written as $V_{\ptc}Z^{\vect{z}}$, where $V_{\ptc}$ is independent of the measurement results and depends only on the subset $\ptc$, and $Z^{\vect{z}}$ is a local unitary correction that depends only on the measurement results $\vect{a}$. The isometry $V_{\ptc}$ is the adjoint of the quantum circuit that maps $\ket{\vect{x}}$ to $\ket{\vect{x}\cdot G_{\ptc}}$ and can be constructed explicitly, similarly to the construction of the encoding isometry $V$. 
\end{proof}

\subsection{Optimal schemes}
As our goal is to reduce the number of quantum communication channels among the players, we aspire to construct schemes in which the LOCC-assisting subsets $\pt$ are as large as possible. The restriction $|\ptc|\geqslant k$ must be satisfied, as otherwise information is lost and thus there is no way for $\ptc$ to recover the quantum secret faithfully~\cite{PhysRevA.61.042311}. We now show that there exist ``optimal" schemes for which  $|\ptc|=k$, which requires the following lemma.
\begin{lemma}\label{lma1}
Every subset $\ptc$ of size $|\ptc|>n-d$, where $d$ is the distance of the underlying classical $[n,k,d]_q$ code, can fully recover the secret $\ket{\psi}$ by LOCC assistance from its complement $\pt$.
\end{lemma}
\begin{proof}
This follows at once as the distance $d$ of the classical code is
\begin{equation}\label{eqn17}
d=1+\max_r\left[ {\mathrm{rank}(G_{\ptc})=k, \forall\ptc\text{ with }|\ptc|=n-r}\right],
\end{equation}
which is to say that one can arbitrarily remove at most $d-1$ columns from the generator matrix $G$ without changing the rank of the resultant $G_{\ptc}$. Therefore, the maximum $r$ in \eqref{eqn17} is $d-1$, which implies that the minimum size of $\ptc$ has to be at least $n-(d-1)=n-d+1$.
\end{proof}

Lemma~\ref{lma1} implies at once that efficient (in terms of quantum communication) LOCC-recoverable secret-sharing schemes are obtained from classical $[n,k,d]_q$ codes that maximize the distance for fixed $n$ and $k$. Such an example is constituted by the class of maximum distance separable (MDS) codes that achieve equality in the classical Singleton bound \cite{NielsenChuang:QuantumComputation}, 
\begin{equation}\label{eqn18}
n-k=d-1.
\end{equation}

\begin{theorem}\label{thm3}
An MDS classical code $[n,k,n-k+1]_q$ induces a quantum secret-sharing scheme in which every subset $\ptc$ of size $k$ or more can recover the quantum secret by LOCC assistance from its complement $\pt$. Furthermore, if $|\ptc|=k$, the scheme is optimal in terms of the number of quantum communication channels required among the players.
\end{theorem}
\begin{proof}
The proof follows immediately from Lemma~\ref{lma1}.
\end{proof}

\section{Example}\label{sct3}
We illustrate our formalism by a concrete example (simple enough to be worked out by hand).

\emph{Example.} Consider a classical repetition code $[n=3,k=1,d=3]_2$ with generator matrix $G=\begin{pmatrix}1 & 1 & 1\end{pmatrix}$ and note that this is an MDS code. The corresponding classical codewords $000$ and $111$ are embedded into two quantum states, $\ket{000}$ and $\ket{111}$, respectively. A secret $\ket{\psi}=c(0)\ket{0}+c(1)\ket{1}$ is distributed to three players as $\ket{\Psi}=c(0)\ket{000}+c(1)\ket{111}$. Theorem~\ref{thm3} implies that any subset $\ptc$ of size $|\ptc|\geqslant 1$ can recover the secret by requiring the players in $\pt$ to perform measurements in the $\{\ket{+},\ket{-}\}$ basis and then sending the measurement results back to $\ptc$. 

Without loss of generality assume that players 1 and 2 perform measurements, with measurement results $a_1\in \FF_2$ and $a_2\in \FF_2$, respectively. The resultant state on the third player is
\begin{equation}\label{eqn19}
\ket{\Psi}_{\{3\},\vect{a}}=c(0)\ket{0}+(-1)^{a_1\oplus a_2}c(1)\ket{1},
\end{equation}
where $\oplus$ denotes addition mod $2$.
Whenever $a_1$ and $a_2$ have the same parity, the third player does not have to do anything. When $a_1$ and $a_2$ are different, then the third player has to apply a $Z$ operator to remove the phase in \eqref{eqn19}. The combined effect can be achieved by player 3 applying the operator $Z^{a_1\oplus a_2}$.

Applying directly our formalism, we have $\pt=\{1,2\}$, $\ptc=\{3\}$, $G_{\pt}=\begin{pmatrix}1 & 1\end{pmatrix}$ and $G_{\ptc}=\begin{pmatrix}1\end{pmatrix}$. Using \eqref{eqn10} we can write the resultant state $\ket{\Psi}_{\{3\},\vect{a}}$ after the measurement performed by the subset $\pt$ in exactly the same form as \eqref{eqn19}.
The operator $Z^z$ the player $\ptc=\{3\}$ has to apply can be found using \eqref{eqn14}, which yields
\begin{equation}\label{eqn20}
z = \begin{pmatrix}1 & 1\end{pmatrix}\cdot\begin{pmatrix}a_1 & a_2\end{pmatrix}^T=a_1\oplus a_2,
\end{equation}
as shown below \eqref{eqn19}. 
This example was first described in the seminal paper of Hillery \textit{et al} \cite{PhysRevA.59.1829}. \textcolor{black}{Their scheme, however, is a particular instance of our general formalism.}

The above scheme can be generalized at once to $n>3$ generalized GHZ states over larger alphabets by using a classical MDS repetition code $[n,1,n]_q$ over $\FF_q$ with generator matrix $G=\begin{pmatrix}1 & 1 & \ldots & 1\end{pmatrix}$. The $\{\ket{\pm}\}$ measurement basis is now replaced by the Fourier basis $\{\ket{\bar x}\}_{x\in\FF_q}$. For a faithful decoding, the $n$th player must apply the operator $Z^{z}$, with
$z = \bigoplus_{i=1}^{n-1}a_i$, where the sum is taken over the elements of $\FF_q$.

\section{Conclusion}\label{sct4}
In summary, we have developed a qudit quantum secret-sharing protocol in which we reduce the quantum communication overhead among the players by enabling some quantum channels to be replaced by classical ones. Our scheme is based on embedding a classical linear code into a quantum code, then using the latter for the actual construction of the protocol. The size of the LOCC-assisting subsets is determined entirely by the error-correcting properties of the classical code. 

\textcolor{black}{The only quantum resource our protocol uses is quantum communication between the dealer and the players, or, equivalently, shared entanglement between the dealer and the players (the latter can be used to simulate quantum channels), and also quantum communication among the players in the subset $\ptc$ (that must recover the quantum secret). The LOCC-assisting subset $\pt$ is only performing local measurements then transmitting the measurement result to its complement. Our scheme is useful in any physical scenario where reducing the amount of expensive quantum communication is essential, such as secret sharing between players with restricted quantum communication ability, for example schemes involving satellites.}
 
As quantum secret-sharing schemes are a form of quantum error-correction, our results represent a first step towards attacking the challenging problem of minimizing the amount of quantum communication needed for decoding the latter. 

\section*{Acknowledgments}
We acknowledge financial support from the Natural Sciences and Engineering Research Council (NSERC) of Canada and from the Pacific Institute for the Mathematical Sciences (PIMS). B.C.S. acknowledges support from CIFAR and China's Thousand Talents Program.


%

\end{document}